\documentclass[pra,twocolumn,floatfix,notitlepage,groupedaddress]{revtex4-1}

\usepackage[hidelinks,colorlinks=false,citecolor=black,linkcolor=black]{hyperref}
\usepackage{mathtools}
\usepackage{subcaption}
\usepackage{amssymb,amsmath,amstext,amsthm}
\usepackage{graphicx}
\usepackage{epstopdf}
\usepackage{color}
\usepackage{bm}
\usepackage{appendix}
\usepackage[T1]{fontenc}
\usepackage{bbm}
\usepackage{latexsym}

\usepackage{algorithm,algorithmic}

\usepackage{float}
\usepackage{verbatim}
\usepackage{lipsum}
\usepackage{braket}
\usepackage{ulem}

\newcommand{\del}[0]{\partial}

\newtheorem{theorem}{Theorem}

\begin{document}

\title{Quantum radar with unreflected photons}
\author{T.J.\,Volkoff}
\affiliation{Theoretical Division, Los Alamos National Laboratory, Los Alamos, NM, USA.}

\begin{abstract}
Two descriptions are introduced and analyzed for a reflectivity estimation and detection scheme that does not involve measurement of photons scattered by the target. One description, provided by the Hamiltonian dynamics of the full transmitter/receiver optical system, incurs an exponential cost in transmitter intensity for a given estimation sensitivity but is linearly improved with the intensity of the thermal background. The other description, based on optical quantum circuits, exhibits sensitivity around a factor of 1/2 of the optimal entanglement-assisted scheme, but incurs an inverse linear reduction in sensitivity with increasing thermal background. The results have applications for the design of optically active receivers based on combining echo-seeded spontaneous parametric downconversion and induced coherence due to photon indistinguishability.
\end{abstract}
\maketitle

\section{\label{sec:intro}Introduction}

A quantum radar system is a quantum information processing protocol that utilizes a source of nonclassical or entangled states of the electromagnetic field to obtain information about a target object (e.g., rough information like its presence or absence as in quantum illumination \cite{shapiro}, or detailed information like its phase space trajectory). Like classical radar, quantum radar schemes can be active or passive, and have a variety of transmitter/receiver structures \cite{skolnik}. The distinguishing feature of quantum radar systems compared to classical counterparts is the use of optically nonclassical or entangled states. Importantly, the nonclassicality or entanglement can be a global property of a multimode state in which some modes are transmitted to target while the complementary modes are kept locally in a phase-sensitive quantum memory for later interference with the reflected echo \cite{lloyd}. A quantum radar system can also make use of entanglement and nonclassicality generated at the receiver component, a technique which has led to near-optimal receivers for Gaussian quantum illumination \cite{PhysRevA.80.052310}. However, entanglement and nonclassicality do not exhaust the quantum properties of optical systems that may be employed for quantum information processing. Recently, it has become clear that by aligning multiple SPDC crystals in a spatial configuration in which downconverted modes are aligned, it is possible to estimate optical phases by measuring photons that never passed through the phase shifting medium \cite{Lemos2014,leecho,PhysRevLett.131.033603,PhysRevA.92.013832}. Applications of this \textit{induced coherence} technique, which relies on photon indistinguishability, include noise-resistant imaging schemes. Looking beyond the setting of optical phase estimation, the present work analyzes the simplest quantum radar system that combines transmitter entanglement and induced coherence at the receiver for estimation or detection of a reflective optical element. Compared to previously studied classical radar or quantum radar systems  which exhibit reflectivity estimation performance scaling as $O({N_{S}\over N_{B}})^{-1}$ in the quantum Cram\'{e}r-Rao bound when transmitting $N_{S}$ photons per pulse in an optical background of $N_{B}$ photons, our main results show that induced coherence radar systems measuring unreflected photons can exhibit a drastically different scaling $O(N_{B}\log N_{S})^{-1}$. The overall dependence provides the possibility of operating regimes in which these novel radar systems surpass commonly cited limitations of target detection with Gaussian states \cite{PhysRevLett.101.253601,Volkoff_2024}.

\begin{figure}[tbh]
  \centering
  \includegraphics[scale=0.7]{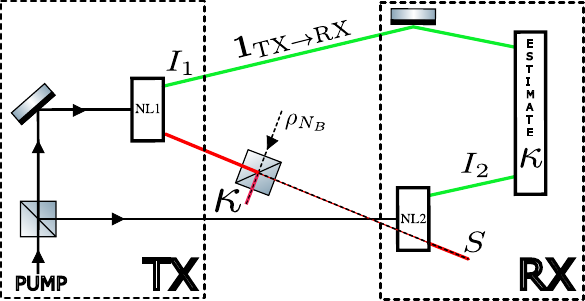}
  
  \caption{ \centering Quantum radar scheme based on the Zou-Wang-Mandel effect. The effective optical diagram has transmitter (TX) and receiver (RX) components. $I_{1}$ is a quantum memory register that appears also in analyses of quantum illumination. The $I_{2}$ register is also entangled with the $S$ mode due to the SPDC process in crystal NL2. The setup can be described by circuit-based dynamics (Model 1) or Hamiltonian dynamics (Model 2). Although the $S$ mode could be measured, we focus on the quantum information about $\kappa$ contained in the unreflected register $I_{1}I_{2}$. }\label{fig:zwmfig}
\end{figure}

In the non-perturbative framework of Gaussian states and operations it is not immediately obvious how to describe systems of SPDC crystals in which optical modes are indistinguishable \cite{PhysRevA.109.023704}. One approach is provided by the example of the $SU(1,1)$ interferometer, which utilizes two $\pi$-phase shifted SPDC crystals to obtain Heisenberg scaling for sensitivity of optical phase shift estimation \cite{PhysRevA.33.4033}. The description of the $SU(1,1)$ interferometer was formulated by direct analogy to the Mach-Zehnder interferometer (i.e., the $SU(2)$ interferometer) by replacing the optical beamsplitters with $SU(1,1)$ operators describing photon number non-conserving (i.e., active) optical elements such as four-wave mixers or SPDC sources \cite{PhysRevA.33.4033}. This approach motivates an optical quantum circuit description of networks of aligned SPDC processes. On the other hand, care must be taken when utilizing quantum circuits to describe dynamics of nonlinear optical systems, because quantum circuits do not necessarily correspond to Hamiltonians that are quadratic in the photon creation and annihilation operators \cite{burgarth2024central}. Even if the quantum circuit can be written as time-evolution under a quadratic Hamiltonian, the mode couplings in the resulting Hamiltonian may not correspond to physical processes in the given system. With these facts in mind, the present work considers two models for the quantum radar system in Fig. \ref{fig:zwmfig}, one circuit-based and one based on the quadratic Hamiltonian dynamics of the full transmitter/receiver system. Common to these models are the following components: 1. two phase-coherent SPDC processes, at the transmitter and receiver respectively, 2. detection of unreflected photons in the modes $I_{1}$ and $I_{2}$, and 3. alignment of the optical echo from the transmitter with a downconversion mode in the receiver which is responsible for induced coherence in the system.

\section{Reflectivity sensing}

We first define and analyze the classical transmitter scheme that will be the standard for comparison for the reflectivity estimation schemes developed in this work. Consider a protocol which distributes a fixed intensity $N_{S}$ to the degrees of freedom (i.e., quadratures) of a single-mode, optically classical Gaussian state which is used to probe a reflective target. Specifically, the transmitter state has the form $\rho_{z,N_{\text{th}}}=D(z)\rho_{0,N_{\text{th}}}D(z)^{\dagger}$ where
\begin{equation}
    \rho_{0,N_{\text{th}}}:={1\over N_{\text{th}}+1}\sum_{n=0}^{\infty}\left( N_{\text{th}}\over N_{\text{th}}+1\right)^{n}\ket{n}\bra{n}
\end{equation}
is a single-mode thermal state with intensity $N_{\text{th}}$ and, for $z\in \mathbb{R}^{2}$, the displacement operator is defined by $D(z):=e^{iR\Omega z}$ with $R=(q,p)$ the row vector of canonical operators, and $\Omega := \begin{pmatrix}
    0&1\\-1&0
\end{pmatrix}$ the symplectic form on $\mathbb{R}^{2}$. The mean vector and covariance matrix of the probe state are respectively
\begin{align}
    m_{\rho_{z,N_{\text{th}}}}&:= \langle R\rangle_{\rho_{z,N_{\text{th}}}} = (z_{1},z_{2}) \nonumber \\
    \Sigma_{\rho_{z,N_{\text{th}}}}&:= {1\over 2}\langle [(R-m_{\rho_{z,N_{\text{th}}}})^{\intercal},R-m_{\rho_{z,N_{\text{th}}}}] \rangle_{\rho_{z,N_{\text{th}}}} \nonumber \\
    &= \left( N_{\text{th}}+{1\over 2}\right)\mathbb{I}_{2}
    \label{eqn:mvcov}
\end{align}
with total signal energy $N_{S} = {\Vert z\Vert^{2}\over 2} + N_{\text{th}}$.
We consider the channel describing reflection from the target to be a bosonic thermal attenuator  $\mathcal{N}_{\kappa,N_{B}}$ which enacts the transformation $\rho_{z,N_{\text{th}}} \mapsto \mathcal{N}_{\kappa,N_{B}}(\rho_{z,N_{\text{th}}})$, thereby leaving the mean vector invariant and mapping the covariance matrix to 
\begin{align}
    \Sigma_{\mathcal{N}_{\kappa,N_{B}}(\rho_{z,N_{\text{th}}})}&= {1\over 2}(2\kappa N_{\text{th}} + 2(1-\kappa)N_{B}+1)\mathbb{I}_{2}.
    \label{eqn:mvcov2}
\end{align}

The parameter $\kappa$ scales with the inverse square of the distance from transmitter to target and linearly with the area of the receiving aperture, according to the classical radar equation \cite{skolnik,PhysRevResearch.2.023414}. The reflected probe state can equivalently be described using the Stinespring form of $\mathcal{N}_{\kappa,N_{B}}$
\begin{equation}
\mathcal{N}_{\kappa,N_{B}}(\rho_{z,N_{\text{th}}}):= \text{tr}_{E}\left[ U_{\kappa}\rho_{z,N_{\text{th}}}\otimes \rho_{0,N_{B}}U_{\kappa}^{\dagger}\right]
\label{eqn:thermalnoise}
\end{equation}
where $U_{\kappa}=e^{i\theta(\kappa)(a^{\dagger}b+h.c.)}$ with $\theta(\kappa):=\cos^{-1}\sqrt{\kappa}$ is a unitary beamsplitter acting on the probe and environment modes  (we have introduced the the creation and annihilation operators of mode $j$ according to $a_{j}^{\dagger} := {q_{j}-ip_{j}\over \sqrt{2}}$). The value $\kappa=0$ corresponds to the signal information being lost, and the signal mode being indistinguishable from the thermal environment. Note that the Stinespring form in (\ref{eqn:thermalnoise}) implies, by monotonicity of QFI under partial trace, that some information about the parameter $\kappa$ is lost to the environment, regardless of the environment temperature. For example, when $N_{B}=0$, the quantum Fisher information of the pure system-plus-environment state $U_{\kappa}\rho_{z,0}\otimes \ket{0}\bra{0}_{E}U_{\kappa}^{\dagger}$ is ${\Vert z\Vert^{2}\over 2\kappa(1-\kappa)}$, whereas the quantum Fisher information of  $\mathcal{N}_{\kappa,0}(\rho_{z,0})$ is ${\Vert z\Vert^{2}\over 2\kappa}$.  

Throughout this work, we utilize a single basic formula for the QFI for $\kappa \in (0,1)$
\begin{align}
    \mathcal{F}(\kappa)&= -4\del_{\eta}^{2}\sqrt{F(\rho_{\kappa},\rho_{\eta})} \big\vert_{\eta=0}
    \label{eqn:qfiform}
\end{align}
where $F(\rho_{\kappa},\rho_{\eta}):= \left(\text{tr}\sqrt{\sqrt{\rho_{\kappa}}\rho_{\eta}\sqrt{\rho_{\kappa}}} \right)^{2}$ is the quantum fidelity which has a known expression for Gaussian state arguments using symplectic invariants \cite{PhysRevA.86.022340}. Again restricting to the case of Gaussian states, the required second derivative can be simplified considerably using the Williamson decomposition of the respective covariance matrices (with a regularization required in the case that some symplectic eigenvalues are equal $1/2$, i.e., in the case that there are pure subsystems) \cite{saffuent}. The formula (\ref{eqn:qfiform}) expresses the relation between a distinguishability-based metric on the state space and the Bures distance \cite{PhysRevLett.72.3439,HUBNER1992239,PhysRevA.95.052320,zhou2019exact}.

\begin{table*}[t]
\centering
\begin{tabular}{p{0.2\linewidth}p{0.12\linewidth}p{0.25\linewidth}p{0.19\linewidth}}
\hline
 Transmitter model & Ent./Cla./Gau. & \centering QFI & Ref. \\
\hline
\hline
Coherent state& No/Yes/Yes & \centering $N_{S}\over \kappa$ & \cite{PhysRevLett.98.160401,9962776}
\\
Best single-mode\\
Gaussian& No/No/Yes & \centering ${N_{S}-O(\kappa)\over \kappa(1-\kappa)}$  & \cite{PhysRevLett.98.160401}
\\
Fock & No/No/No & \centering ${N_{S}\over \kappa(1-\kappa)}$ & \cite{PhysRevA.79.040305,PhysRevResearch.2.033389} (see Eq.(\ref{eqn:transf})) \\
Two-mode squeezed & Yes/No/Yes & \centering ${N_{S}\over \kappa(1-\kappa)}$& \cite{PhysRevResearch.6.013034}
\\
Unreflected photons \\(Model 2) & Yes/No/Yes & \centering ${\left( \log\left( \sqrt{N_{S}+1} + \sqrt{N_{S}} \right) \right)^{2} \over 2\kappa(1-\kappa)}$ & This work\\
Unreflected photons \\(Model 1) & Yes/No/Yes & \centering ${N_{S}(1+(1-\kappa)N_{S})\over \kappa(1-\kappa)(2+(2-\kappa)N_{S})}$ & This work\\
\hline
\end{tabular}
\caption{Reflectivity estimation with a noiseless transmitter with transmitted intensity $N_{S}$. Ent./Non./Gau. stands for Entangled/Classical/Gaussian. }
\label{tab:noiselessfisher}
\end{table*}

It is important to note that although estimation of the $\kappa$ parameter does not fall under the commonly used shift model of parameter estimation in which a parameter $\theta$ is imprinted on a $\theta$-independent probe state $\ket{\psi}$ according to $\ket{\psi}\mapsto e^{-i\theta A}\ket{\psi}$ for some self-adjoint operator $A$ \cite{holevo}, it is related to a shift model by change of variables. Specifically, let $A=a_{1}^{\dagger}a_{2}+h.c.$. Then taking $\theta(\kappa) = \cos^{-1}\sqrt{\kappa}$ gives the quantum state-valued function of $\kappa$ that we seek to analyze. Fixing this relation between the parameters, the general relations
\begin{equation}
    \mathcal{F}(\kappa) = {\mathcal{F}(\theta(\kappa)) \over 4\kappa(1-\kappa)} = {\mathcal{F}(\sqrt{\kappa}) \over 4\kappa}
    \label{eqn:transf}
\end{equation}
which relate the QFI on the state manifold parametrized by $\kappa$ to the QFI on the state manifold parametrized by $\theta$ to the QFI on the state manifold parametrized by $\sqrt{\kappa}$, hold regardless of whether the shift model is used to parametrize the state.
Relation (\ref{eqn:transf}) is often useful because in the shift model, $\mathcal{F}(\theta)$ is independent of $\theta$, so that the analytic behavior of $\mathcal{F}(\kappa)$ can be determined immediately from the denominator of (\ref{eqn:transf}).   

The QFI for the state in (\ref{eqn:thermalnoise}) is
\begin{align}
    \mathcal{F}(\kappa) &= {(N_{\text{th}}-N_{B})^{2}\over ((1-\kappa)N_{B} + \kappa N_{\text{th}}+1)((1-\kappa) N_{B} + \kappa N_{\text{th}})} \nonumber \\
    &{} + {\Vert z\Vert^{2}\over 2\kappa(2N_{\text{th}}\kappa + 2N_{B}(1-\kappa) + 1)}.
    \label{eqn:qfi1}
\end{align}
 Note that with $N_{\text{th}}= 0$, i.e., a pure displaced vacuum transmitter, our result for the QFI is not in agreement with Ref.\cite{9962776}. That is because we have not disregarded the ``shadow effect'', which accounts for information about the reflectivity $\kappa$ that is present due to scattering of the thermal environment by the target. Mathematically, neglecting the shadow effect amounts to taking $N_{B}\mapsto {N_{B}\over 1-\kappa}$ in the covariance matrix (\ref{eqn:mvcov2}), so that at $N_{\text{th}}=0$, the covariance matrix is $\kappa$-independent. Of course, at $N_{\text{th}}=N_{B}=0$, the result in (\ref{eqn:qfi1}) and Ref.\cite{9962776} are in agreement.

In the optical domain, $N_{B}\ll 1$ and one finds that subject to the total transmitter energy constraint ${\Vert z \Vert^{2}\over 2}+N_{\text{th}} = N_{S}$, the QFI (\ref{eqn:qfi1}) is ${N_{S}\over \kappa}$ at the lowest order ($N_{B}^{0}$ and $\kappa^{-1}$). Therefore, the distribution of energy between coherent lasing and incoherent thermal noise is immaterial to the sensitivity obtainable with the probe state $\rho_{z,N_{\text{th}}}$. In the opposite parameter domain $N_{B}\gg N_{S}$, the QFI is ${\Vert z \Vert^{2}\over 4\kappa N_{B}}$ to order $\kappa^{-1}$ and $N_{B}^{-1}$, which indicates that all intensity should be distributed to coherent lasing. The same distribution of intensity is optimal for the domain $1\ll N_{B}\ll N_{S}$ (e.g., strong microwave probe signal), for which the QFI is ${\Vert z\Vert^{2}\over 2\kappa(2N_{B}+1)}$ at order $\kappa^{-1}$.

 Although (\ref{eqn:qfi1}) is exact, the setting of quantum radar is described by $\kappa \ll 1$. In this domain, the leading order contribution to $\mathcal{F}(\kappa)$ is $O(\kappa^{-1})$. Therefore, we compare the performance of any two estimation schemes by the limiting value of the ratio of their QFI values as $\kappa\rightarrow 0$. For example, consider the two mode squeezed state transmitter 
 \begin{equation}
\ket{\psi_{\text{TMSS}}}_{SI}:=(U_{\text{TM}}(g))_{SI}\ket{0}_{S}\ket{0}_{I}
 \end{equation}
 in which the transmitted mode $S$ is entangled with a local quantum memory mode $I$ by the unitary two-mode squeezing operation $(U_{\text{TM}}(g))_{SI}:= e^{g(a_{S}^{\dagger}a_{I} - h.c.)}$, resulting in both modes having expected photon number $N_{S}=\sinh^{2}g$. The thermal attenuator $\mathcal{N}_{\kappa,N_{B}}^{(S\rightarrow S)}$ acts on the $S$ mode according to (\ref{eqn:thermalnoise}), producing the output state $\mathcal{N}_{\kappa,N_{B}}^{(S\rightarrow S)}(\psi_{\text{TMSS}})$ on the $SI$ register. The QFI of the resulting state  is given by
 \begin{align}
     \mathcal{F}(\kappa)&= {N_{S}(N_{S}+1)-\kappa(N_{S}^{2}-N_{B}(2N_{S}+1))\over \kappa(1-\kappa)(1+(1-\kappa)(N_{S}+N_{B}+2N_{S}N_{B})} \nonumber \\
     &=  {N_{S}(N_{S}+1) \over 1+N_{B}+N_{S} + 2N_{B}N_{S}}\kappa^{-1}+O(\kappa^{0})
     \label{eqn:tmssqfi}
 \end{align}
 with a $10\log_{10}(1-\kappa)$ dB amplification from the coherent state when $N_{B}=0$. Note that the $I$ mode alone has no sensitivity to $\kappa$, and that when $N_{B}\ge {\kappa\over 1-\kappa}$ the state of $SI$ after the thermal attenuator is both unentangled (by checking the separability condition that the symplectic eigenvalues of the covariance matrix of the partially transposed state are at least $1/2$ when this condition holds \cite{PhysRevA.70.022318}) and classical (by checking that the eigenvalues of the covariance matrix are at least $1/2$ when this condition holds). The interesting physical point about the two-mode squeezed state transmitter is that for $N_{\text{th}}=0$ and $N_{S}\in \mathbb{N}_{\ge 0}$, it matches the QFI of a single mode Fock state $\ket{N_{S}}$ (see Table \ref{tab:noiselessfisher}), which is the optimal single mode transmitter. This shows that non-Gaussianity is not necessary for optimal performance. In the present section, the dependence of the QFI on the transmitter mode energy and the energy of the thermal environment mode scales as $O({N_{S}\over N_{B}})$ for all examples. This scaling will not occur for all of the quantum radar models that we introduce in Section \ref{sec:ttt} with the aim of describing the transmitter/receiver system in Fig. \ref{fig:zwmfig}. In Section \ref{sec:ub} we show how general couplings between $S$, $I$, and $E$ result in an upper bound on the QFI that scales linearly with $N_{B}$.

 \section{Upper bound from environment information\label{sec:ub}}
 We do not delve into the problem of determining how closely the noise-to-signal ratio for readouts such as intensity, homodyne, or heterodyne measurement come to saturating the QFI, noting only that because we are considering estimation of a single parameter, the QFI is achievable by the measurement of the symmetric logarithmic derivative operator \cite{PhysRevLett.72.3439,FUJIWARA1995119}. We note that, despite the fact that many non-classical Gaussian and non-Gaussian states have QFI which exhibits Heisenberg scaling $\mathcal{F}(\kappa)= {O(N_{S}^{2})\over \kappa(1-\kappa)}$ in a closed system where $U_{\kappa}$ acts on two modes $SI$ of an input state, the upper bound on the QFI in Theorem 1 of Ref.\cite{PhysRevLett.118.070803} precludes Heisenberg scaling for states of the form $\mathcal{N}^{(S\rightarrow S)}_{\kappa,N_{B}}(\psi_{SI})$, which includes the examples in the present section. However, for quantum radar schemes which are not described by that probe state, e.g., schemes with optically active receivers, the calculation of the QFI in the proof of that Theorem is not generally valid. For example, it is possible that the QFI increases linearly with $N_{B}$ instead of decreasing like $N_{B}^{-1}$. The theorem below states that if some information from the environment is coupled back to the receiver, then $O(N_{S}N_{B})$ scaling provides an upper bound. The upper bound is saturated to within a constant factor by using an SPDC transmitter and a receiver that mixes the $I$ and $E$ mode on 50:50 beamsplitter.
 \begin{theorem}
 \label{thm:ooo}
     For any state of the $SI$ register of the form
     \begin{equation}
         \mathrm{tr}_{E}\left[ V_{ISE}(U_{\kappa})_{SE}\psi_{IS}\otimes (\rho_{0,N_{B}})_{E} (U_{\kappa})_{SE}^{\dagger}V_{ISE}^{\dagger} \right]
         \label{eqn:mas}
     \end{equation}
     with $\ket{\psi}_{IS}$ a pure state such that $\langle a^{\dagger}_{S}a_{S}\rangle_{\psi_{IS}} = N_{S}$, and $V_{ISE}$ an unparametrized unitary, the QFI satisfies \begin{equation}\mathcal{F}(\kappa)\le {N_{S}+N_{B}+2N_{S}N_{B}\over \kappa(1-\kappa)}.\end{equation} If $\ket{\psi}_{IS}=\ket{\psi_{\mathrm{TMSS}}}_{IS}$ with $g=\log(\sqrt{N_{S}+1}+\sqrt{N_{S}})$, and $V_{ISE}=e^{i{\pi\over 4}(a_{I}^{\dagger}a_{E} + h.c.)}\otimes \mathbb{I}_{S}$, then
     \begin{equation}
         \mathcal{F}(\kappa) \sim {N_{S}+N_{B}+2N_{S}N_{B}\over 2\kappa}
     \end{equation}
     in the limit $\lim_{N_{B}\rightarrow \infty}\lim_{\kappa\rightarrow 0}$.
 \end{theorem}
 \begin{proof}
    Define the state
     \begin{align}
         &{}\ket{\phi_{\kappa}}_{ISE_{1}E_{2}}:= V_{ISE_{1}}(U_{\kappa})_{SE_{1}}\sum_{k}\sqrt{q_{k}}\ket{\psi}_{IS}\ket{k}_{E_{1}}\ket{k}_{E_{2}}\nonumber
     \end{align}
     with $q_{k}:={1\over N_{B}+1}\left( {N_{B}\over N_{B}+1}\right)^{k}$, which is a purification of (\ref{eqn:mas}). Call the QFI of the parametrized state $\ket{\phi_{\kappa}}_{ISE_{1}E_{2}}$ by $\mathcal{G}(\kappa)$. That $\mathcal{F}(\kappa)\le \mathcal{G}(\kappa)$ follows from the monotonicity of QFI under partial trace. The operation $V_{ISE}$ has no effect on $\mathcal{G}(\kappa)$ because it is unparametrized. Writing $\ket{\psi}_{IS}=\sum_{m,n}c_{m,n}\ket{m}_{I}\ket{n}_{S}$ in full generality, one obtains from (\ref{eqn:transf}) that
     \begin{align}
         4\kappa(1-\kappa)\mathcal{G}(\kappa)&= 4\text{Var}_{\ket{\phi_{0}}}\left( a_{S}^{\dagger}a_{E} + h.c. \right) \nonumber \\
         &= 4\sum_{m,n,k}\vert c_{m,n}\vert^{2}q_{k}(n(k+1) + k(n+1)) \nonumber \\
         &= 4(N_{S}+N_{B} + 2N_{S}N_{B}).
     \end{align}

     The two-mode squeezed state $\ket{\psi_{\text{TMSS}}}_{SI}$ with $g=\log(\sqrt{N_{S}+1}+\sqrt{N_{S}})$ has energy $N_{S}$ in both the $S$ and $I$ modes. Since it is Gaussian, the covariance matrix of (\ref{eqn:mas}) is straightforwardly obtained by symplectic transformations. For two mode Gaussian states, the formula (\ref{eqn:qfiform}) can be written \cite{saffuent}
     \begin{equation}
         \mathcal{F}(\kappa)= {\del_{\kappa '}^{2}\Delta-2\del_{\kappa '}^{2}(\sqrt{\Gamma}+\sqrt{\Lambda})\over \sqrt{\Gamma}+\sqrt{\Lambda} - 1} \Big\vert_{\kappa ' =\kappa}
     \end{equation} where  $\Delta, \Gamma, \Lambda$ are functions of  $N_{S},N_{B},\kappa,\kappa'$ defined in Ref.\cite{PhysRevA.86.022340}. Taking the series expansions with respect to $\kappa$ about 0  and $N_{B}$ about $\infty$, in that order, gives the stated asymptotic.
 \end{proof}
 Comparing $\mathcal{G}(\kappa)$ for $N_{B}=0$ to the optimal performance of a single mode transmitter in Table \ref{tab:noiselessfisher}, this theorem can also be regarded as a ``no-go'' result for entanglement-assisted estimation of $\kappa$ for the class of quantum radar protocols described by state (\ref{eqn:mas}) at the receiver. 
 
\section{Induced coherence for quantum radar\label{sec:ttt}}

We now move to the main problem of the present paper, which is to address the possibility of utilizing SPDC sources at the both the transmitter and receiver in such a way that target-scattered photons from the first SPDC are indistinguishable from photons created at the same wavelength in the receiver SPDC. An optical diagram showing the transmitter and receiver structure is shown in Fig. \ref{fig:zwmfig}. Note that unlike an $SU(1,1)$ interferometer, there is not a $\pi$ phase shift of a pump between the two SPDC sources. Mathematically, this allows one to describe both SPDC sources by the same unitary operation, defined by positive nonlinearity parameter $g>0$.

The first model (Model 1) for describing this system treats the reflected light as seeding the second SPDC source NL$2$ at the receiver in Fig. \ref{fig:zwmfig}. The channel maps the $I_{1}I_{2}$ vacuum to
\begin{align}
    \text{tr}_{SE}\left[ (\mathcal{U}_{\text{TM}})_{SI_{2}}(\mathcal{U}_{\kappa})_{SE}(\mathcal{U}_{\text{TM}})_{SI_{1}} \left( \ket{0}\bra{0}_{SI_{1}I_{2}}\otimes (\rho_{0,N_{B}})_{SE} \right) \right]
    \label{eqn:model2}
\end{align}
where the calligraphic $\mathcal{U}$ just gives the usual Schr\"{o}dinger action $U(\cdot)U^{\dagger}$ of the specified unitary on the specified registers. We refer to (\ref{eqn:model2}) as Model 1. Note that the received state is on the register $I_{1}I_{2}$, which contains photons that were never transmitted to or reflected from the target. However, $I_{1}I_{2}$ contains information about $\kappa$ due to entanglement with $S$. The register $I_{1}$ can be interpreted as a quantum memory register that stores photons entangled with the transmitted (signal) mode $S$ from the transmitter SPDC process, whereas the register $I_{2}$ contains photons newly created at the receiver SPDC process.  This description coincides with the first non-perturbative descriptions of the Zou-Wang-Mandel experiment \cite{BELINSKY1992303,WISEMAN2000245}, and was later utilized to analyze the ultimate sensitivity of an interferometry scheme based on path identity induced quantum coherence \cite{Plick_2010}. The perturbative description of this experiment \cite{PhysRevA.41.1597,PhysRevA.44.4614,PhysRevA.92.013832}, while convenient in terms of detailed descriptions of temporal, spatial, and polarization structure, is not sufficient to distinguish the two models for quantum radar with unreflected photons that we consider in the present paper. Therefore, maintaining a general positive value of the SPDC nonlinearity $g$, the QFI for Model 1 described by (\ref{eqn:model2}) is calculated from (\ref{eqn:qfiform}) to be
\begin{align}
    \mathcal{F}(\kappa)&:= {N_{S}\over \kappa(1-\kappa)}{G_{1}(N_{S},N_{B},\kappa) \over G_{2}(N_{S},N_{B},\kappa)} \nonumber \\
    G_{1} &=  (N_{S}+1)(1+N_{S}+N_{S}N_{B})  \nonumber \\
    &{}  +\kappa (N_{B}-N_{S})(1+N_{S}+N_{B}+2N_{S}N_{B})  \nonumber \\
    &{}  +\kappa^{2}N_{B}(N_{S}^{2}-N_{B}(2N_{S}+1))  \nonumber \\
    G_{2} &=(N_{S}+1 + (1-\kappa)N_{B}N_{S}) \nonumber \\
    &{} \cdot (2+(2-\kappa)N_{S} + (1-\kappa)N_{B}(2N_{S}+1))
\end{align}
with small $\kappa$ expansion
\begin{equation}
    \mathcal{F}(\kappa)= {N_{S}(N_{S}+1) \over 2+N_{B}+2N_{S} + 2N_{B}N_{S}}\kappa^{-1}+O(\kappa^{0}).
    \label{eqn:model2expansion}
\end{equation}
Note that unlike the two-mode squeezed state transmitter, in which modes $SI$ are used estimate $\kappa$, the modes $I_{1}I_{2}$ that are utilized to estimate $\kappa$ in Model 1 were never reflected from the target. Despite this fact, for $N_{B}=0$ and $\kappa\ll 1$ one finds from (\ref{eqn:model2expansion}) that the QFI for Model 1 is half of that for the two-mode squeezed state transmitter in (\ref{eqn:tmssqfi}). Although (\ref{eqn:model2}) would have the form of the state (\ref{eqn:mas}) in Theorem \ref{thm:ooo} if the trace over $S$ were not taken, the $(U_{\text{TM}})_{SI_{2}}$ commutes with $\text{tr}_{E}$, so the bound from Ref.\cite{PhysRevLett.118.070803} is applicable and tighter. If Model 1 is indeed a viable description of Fig. \ref{fig:zwmfig}, one expects that many advantages of quantum imaging with undetected photons (e.g., robustness of estimation in the presence of transmitter-wavelength-specific noise channels) \cite{bl22}, carry over to the setting of quantum radar with unreflected photons.

Besides the concern that Model 1 does not accurately describe the nonlinear optical system due to its circuit structure, it is also notable that increasing the thermal background $N_{B}$ reduces the sensitivity monotonically. However, it is not clear that increasing $N_{B}$ should decrease the sensitivity in the scheme of Fig.\ref{fig:zwmfig}. In particular, the target reflection $U_{\kappa}$ causes thermal light to seed the receiver SPDC, so increasing the thermal background increases the number of photons in the register $I_{1}I_{2}$. These photons carry information about $\kappa$ due to having scattered from the target, so could potentially increase the QFI.

The second model for describing Fig.\ref{fig:zwmfig} exhibits monotonically increasing QFI with respect to $N_{B}$ for $\kappa \ll 1$ and $N_{B}\gg N_{S}$. However, in this description, information about the parameter $\kappa$ only appears by its modification of the SPDC coupling of the $I_{1}$ mode to the $S$ and $E$ modes in the global Hamiltonian. Specifically, the channel maps the $I_{1}I_{2}$ vacuum to
\begin{align}
    &{} \text{tr}_{SE}\left[ e^{-igH_{SI_{1}I_{2}E}}\ket{0}\bra{0}_{SI_{1}I_{2}}\otimes (\rho_{0,N_{B}}(\kappa))_{SE} e^{igH_{SI_{1}I_{2}E}}\right] \nonumber \\
    &{} H_{SI_{1}I_{2}E} = i\left( (\sqrt{\kappa}a_{S}^{\dagger} - i\sqrt{1-\kappa}a_{E}^{\dagger})a_{I_{1}}^{\dagger} + a_{S}^{\dagger}a_{I_{2}}^{\dagger} - h.c. \right)
    \label{eqn:model1}
\end{align}
where $(\rho_{0,N_{B}}(\kappa))_{SE}$ is a thermal state of the mode $\sqrt{\kappa}a_{E}^{\dagger} -i\sqrt{1-\kappa}a_{S}^{\dagger}$ with occupation $N_{B}$.
This channel describes the scheme in Fig. \ref{fig:zwmfig} as the unitary time dynamics of a non-linear optics system, with no circuit structure imposed. Although in Fig. \ref{fig:zwmfig}, the thermal state $(\rho_{0,N_{B}}(\kappa))_{SE}$ appears as the output state of a beamsplitter circuit element, physically this state is just an initial condition, namely the state of the $SE$ register when no transmitter is sent. For example, when no target is present ($\kappa=0$), the thermal state becomes $(\rho_{0,N_{B}}(\kappa))_{SE} = (\rho_{0,N_{B}})_{S}$, which expresses the fact that the received $S$ mode is indistinguishable from thermal environment. Alternatively, for $\kappa\neq 0$, the target is present and partially shields the $S$ mode from thermal background. Note that in (\ref{eqn:model1}), the term in the Hamiltonian $a_{S}^{\dagger}(\sqrt{\kappa}a_{I_{1}}^{\dagger} + a_{I_{2}}^{\dagger}) - h.c.$ indicates the induced coherence between $I_{1}$ and $I_{2}$ modes that exists when the target is present due to the assumption of perfect indistinguishability of the reflected echo and receiver SPDC photons created in the $S$ mode.

The transmitter state is also not described by a quantum circuit, but rather by specifying the two relevant SPDC processes in the Hamiltonian $H_{SI_{1}I_{2}E}$: 1. the transmitter SPDC source simultaneously creates an  $I_{1}$ photon and a photon in superposition between $S$ and $E$ modes with the reflectivity $\kappa$ determining the relative amplitude, 2. the receiver SPDC source creates photons in $S$ and $I_{2}$ with the bare amplitude $g$. The covariance matrix of the $I_{1}I_{2}$ subsystem is calculated from (\ref{eqn:model1}) by finding the Heisenberg picture dynamics generated by $H$, calculating the full covariance matrix, then restricting to the $I_{1}I_{2}$ subsystem. At $\kappa=0$, this subsystem has the same total energy $N_{S}(N_{B}+2)$ as Model 1. Using a vectorized expression for the QFI of Gaussian states based on the symmetric logarithmic derivatives \cite{Gao2014} and a computer algebra system, we analytically obtained the leading order $O(\kappa^{-1})$ contribution to the QFI for $N_{B}>0$, with the exact value for $N_{B}=0$ shown in Table \ref{tab:noiselessfisher}. Taking also $N_{S}\gg 1$ one obtains
\begin{equation}
    \mathcal{F}(\kappa)= {\left( (N_{B}+2)g(N_{S})-N_{B} \right)^{2} \over 8\kappa(1+N_{B})}  +\tilde{O}(N_{S}^{-1}\kappa^{-1})+O(\kappa^{0})
    \label{eqn:zwmqfi}
\end{equation}
where $g(N_{S}):= \log\left( \sqrt{N_{S}+1}+\sqrt{N_{S}} \right)$ (from the relation $N_{S}=\sinh^{2}g$ for the energy of the transmitter mode $S$, and with $\tilde{O}$ suppressing polylogarithmic factors in $N_{S}$. Therefore, although the logarithmic dependence on $N_{S}$ is unavoidable, the thermal background $N_{B}$ compensates the QFI. Since the thermal background is an environmental source, one can consider Model 2 to be an active-passive quantum radar system, with signal transmission being necessary for nonzero sensitivity, but with increasing thermal background  providing linearly more information about the reflectivity of the target (Fig. \ref{fig:qfifig}).

\begin{figure*}[t]
    \centering
    \includegraphics[scale=0.6]{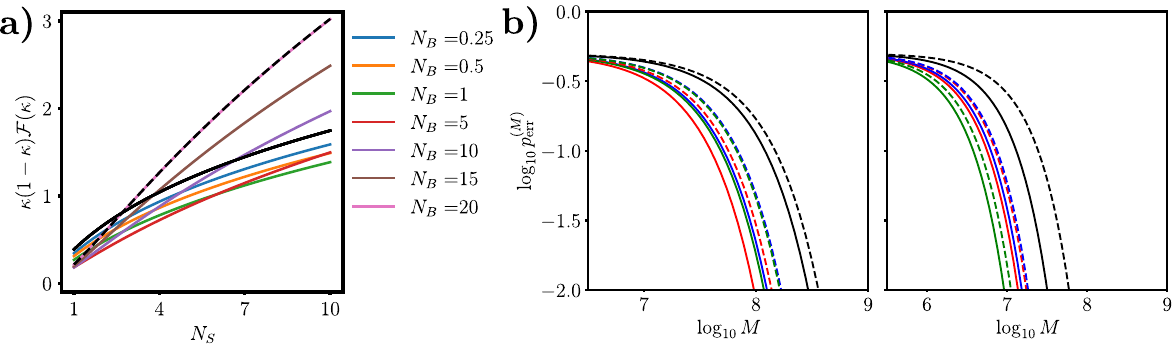}
    \caption{a) QFI scaled by $\kappa(1-\kappa)$ for Model 2 for increasing $N_{B}$, with $\kappa = 10^{-3}$. Black line is the analytical QFI for $N_{B}=0$, dashed black line is the $\kappa \rightarrow 0$ asymptotic analytical QFI (this function fits the other curves also, but is not shown). b) Optimal error probability using the QCE for $N_{B}=20$, $\kappa = 10^{-4}$ and $N_{S}=10^{-2}$ (left) or $N_{S}=10^{-1}$ (right). Dashed lines show upper bounds from (\ref{eqn:uiui}) (note that the asymptotic in (\ref{eqn:zwmqfi}) is not valid for these $N_{S}$; one should use the full expression). Black is coherent state, red is two-mode squeezed state, blue is Model 1, green is Model 2.}
    \label{fig:qfifig}
\end{figure*}

It is interesting that although Model 2 exhibits linearly increasing QFI with $N_{B}$, it does not have the form of the state (\ref{eqn:mas}) in Theorem \ref{thm:ooo}, so is excluded from that result. Clearly, the upper bound in Theorem \ref{thm:ooo} is tailored to quantum circuit-based descriptions of quantum radar schemes. However, we hypothesize that the linear scaling with respect to $N_{B}$ in Model 2 can be explained by the effect of the environment mode on the effective coupling between the $I_{1}I_{2}$ modes in (\ref{eqn:model1}), which physically plays a similar role to the quantum circuit element $V_{ISE}$ in (\ref{eqn:mas}). It is clear that under Model 1, photons from $I_{2}$ never occur simultaneously with photons from $I_{1}$ and $E$.

 We now provide an algebraic argument that is relevant to understanding equivalent nonlinear optical processes for quantum circuits that contain squeezing or SPDC elements. By enumerating the $8\times 8$ matrices defining the Lie algebra $\mathfrak{sp}(8,\mathbb{R})$, one can attempt to numerically solve the equation
 \begin{equation}
    C^{\dagger}\begin{pmatrix}
         \vec{a}\\\vec{a}^{\dagger}
     \end{pmatrix}C = e^{tx + i\sum_{j}\alpha_{j}Q_{j}}\begin{pmatrix}
         \vec{a}\\\vec{a}^{\dagger}
     \end{pmatrix}
 \end{equation}
 for the matrix $x$ and $\alpha_{j} \in \mathbb{R}$, where $C$ is a given circuit $C$ composed of bosonic Gaussian unitaries on $SI_{1}I_{2}E$ and the $Q_{j}$ are $8\times 8$ matrices constituting a system of algebraic charges not contained in $\mathfrak{sp}(8,\mathbb{R})$. By associating a solution matrix $x$ with its bosonic counterpart \cite{perelomov}, and associating $\alpha$ with a charge appended to $\mathfrak{sp}(8,\mathbb{R})$ \cite{burgarth2024central}, a solution allows one to identify a quadratic Hamiltonian that corresponds to $C$ up to a phase shift between the even and odd photon number sectors of the Hilbert space of four optical modes. Although we found that Model 1 can be written as the unitary dynamics generated by a Hamiltonian quadratic in the creation and annihilation operators (i.e., Model 1 is a bosonic Gaussian dynamics, with $Q_{j}=0$), this Hamiltonian involves the SPDC process $a_{E}^{\dagger}a_{I_{2}}^{\dagger}+h.c.$. Such a process is evidently not contained in the optical diagram in Fig. \ref{fig:zwmfig}, which casts some doubt that Model 1 correctly describes the transmitter/receiver system.
 
In addition to comparisons based on the QFI, it is instructive to compare Model 1 and Model 2 in the task of target detection, which we describe by symmetric hypothesis testing. The relevant information theoretic quantity is the minimal probability or error $p_{\text{err}}^{(M)}$ in the task distinguishing $M$ copies of the state of the receiver when $\kappa \neq 0$ from the state of the receiver when $\kappa=0$, and we allow global measurements of the $M$ copies \footnote{In the setting of quantum radar, the copies can be considered as the echos from independent transmissions interrogating the target. A global measurement of the $M$ copies therefore assumes another quantum memory that stores the echos at the receiver. Near-term local measurement schemes have been proposed via sequential detection \cite{PhysRevApplied.20.014030}.}. In the limit of large $M$, the quantum Chernoff theorem identifies the logarithm of this error probability with the quantum Chernoff exponent (QCE) denoted by $\mathcal{C}(\rho_{0},\rho_{\kappa})$ according to $\lim_{M\rightarrow \infty}-{1\over M}\ln p_{\text{err}}^{(M)}(\rho_{0},\rho_{\kappa}) =\mathcal{C}(\rho_{0},\rho_{\kappa})$. Technical definitions of the error probability and QCE can be found in Ref. \cite{wp}, and rigorous analytical calculations of the QCE for coherent state and one-mode and two-mode squeezed state transmitters appear in Ref.\cite{Volkoff_2024}. Although upper bounds on $p_{\text{err}}^{(M)}$ are known in the case of Gaussian states in terms of the symplectic spectra of their covariance matrices \cite{PhysRevA.78.012331,PhysRevA.84.022334}, under certain conditions it is also possible to obtain two-sided bounds in terms of the QFI. Such bounds allow one to relate the QFI values obtained above to the target detection problem. Specifically, in a neighborhood of $\kappa=0$ for which the square root of the quantum fidelity $\sqrt{F(\rho_{\theta={\pi\over 2}},\rho_{\theta})}$ \cite{nielsen} is less than its quadratic approximation, one finds that
\begin{align}
    p_{\text{err}}^{(M)} 
    &\le {1\over 2}e^{-{M\over 2}\left( {\pi\over 2} - \cos^{-1}\sqrt{\kappa}\right)^{2}\lim_{\kappa\rightarrow 0^{+}}\kappa(1-\kappa)\mathcal{F}(\kappa)}.
    \label{eqn:uiui}
\end{align}
To see this, first note that the QFI at $\kappa=0$ is not defined, so to analyze the fidelity at second order, one should consider the receiver states as a function of the angle $\theta = \cos^{-1}\sqrt{\kappa}$, $\theta \in [0,{\pi\over 2}]$ and note that $p_{\text{err}}^{(M)}$ is independent of the parameterization of the hypothesis states. The Fuchs-van de Graaf inequality \cite{761271} then gives
\begin{align}
    p_{\text{err}}^{(M)}(\rho_{0},\rho_{\kappa})& \le {1\over 2}F(\rho_{\theta = \pi/2},\rho_{\theta})^{M/2}\nonumber \\
    &\le {1\over 2}e^{-{M\over 8}\mathcal{F}(\theta)\vert_{\theta = {\pi/2}}(\theta - {\pi\over 2})^{2}}
\end{align}
where the second line used the assumed neighborhood of $\theta=\pi/2$. Using (\ref{eqn:transf}) to reparameterize the state manifold gives (\ref{eqn:uiui}). Numerical values of $\log_{10}{1\over 2}e^{-M\mathcal{C}(\rho_{0},\rho_{\kappa})}$ are shown in Fig. \ref{fig:qfifig}b. The plots show that there are operating regimes of quantum radar in which utilization of only unreflected photons can actually improve the minimal detection error probability. The scaling of the QFI with respect to $N_{B}$ and $N_{S}$ combined with the bound in (\ref{eqn:uiui}) suggests that the error probability for Model 2 can be arbitrarily lower compared to a classical transmitter, two-mode squeezed state transmitter, or the circuit-based Model 1.

\section{Discussion}
Despite its modest advantages compared with other quantum metrology settings such as optical phase estimation, reflectivity estimation is crucial to understanding the performance of quantum radar systems. The results of the present work provide two descriptions of a quantum radar scheme that exploits quantum entanglement at both the transmitter and receiver, along with indistinguishability of the optical echo and newly produced receiver photons, to estimate or detect $\kappa$ using photons than never incur the $R^{-4}$ attenuation associated with propagation along the transmitter-target-receiver path. If SPDC processes are used to produce the entanglement, the measured photons do not even have to be of the same frequency as the transmitted mode. The Hamiltonian dynamics model of Fig. \ref{fig:zwmfig} also exhibits the property of enhanced sensitivity with increasing thermal background. This unusual property is also the source of the reduced detection error compared even to other quantum radar schemes. Inspection of the circuit-based dynamics (\ref{eqn:model2}) and Hamiltonian dynamics (\ref{eqn:model1}) suggests that environment photon-mediated coupling between the $I_{1}$ and $I_{2}$ modes is the source of the qualitative performance differences between Model 1 and Model 2 as radar systems.

In a practical quantum radar scheme based on induced coherence by path alignment, it would be challenging to guarantee indistinguishability of a reflected photon and a spatially aligned photon produced by the receiver SPDC due to the differential phase shift. In other words, it would be a serious challenge to guarantee that the same optical mode $S$ is common to all operations in (\ref{eqn:model2}). A potential solution to this challenge is to utilize frequency comb pumps for the SPDC processes with the aim of ensuring phase coherence of the reflected photons and their indistinguishable counterparts produced at the receiver SPDC \cite{dalprep}. Assuming that the induced coherence can be achieved, future experiments on multiparameter optical phase estimation in the presence of asymmetric transmission \cite{PhysRevA.102.013704,PhysRevA.106.023716,PhysRevA.107.043704} may prove crucial to understanding the correct quantum dynamics. More generally, the present results provide a foundation for understanding quantum information processing applications of more general optical downconversion networks with indistinguishable output modes.

\acknowledgements
The author thanks N. Dallmann, D. Dalvit,  L. Davidovich, and R. Newell for useful discussions.  Los Alamos National Laboratory is managed by Triad National Security, LLC, for the National Nuclear Security Administration of the U.S. Department of Energy under Contract No. 89233218CNA000001.
\bibliography{np.bib}
\bibliographystyle{unsrt}


\end{document}